\documentclass{amsart}
\usepackage[foot]{amsaddr}

\usepackage{fullpage}
\usepackage{amsmath,amssymb,amsthm}
\usepackage{tikz}
\usepackage{hyperref}

\theoremstyle{plain}
\newtheorem{theorem}{Theorem}
\newtheorem{lemma}{Lemma}
\newtheorem{corollary}{Corollary}

\theoremstyle{definition}

\newtheorem{remark}{Remark}
\newtheorem{example}{Example}

\let\leq\leqslant
\let\geq\geqslant

\title{$H$-supermagic labelings for firecrackers, banana trees and flowers}

\author{Rachel Wulan Nirmalasari Wijaya$^1$}
\address{$^1$School of Mathematical and Physical Sciences, University of Newcastle, 2308 NSW, Australia}
\author{Andrea Semani\v{c}ov\'{a}-Fe\v{n}ov\v{c}\'{i}kov\'a$^2$}
\address{$^2$Department of Applied Mathematics and Informatics,  Technical University, Letn\'a 9, Ko\v
  sice, Slovakia}
\author{Joe Ryan$^1$}
\author{Thomas Kalinowski$^1$}

\email{rachel.wijaya@uon.edu.au}
\email{andrea.fenovcikova@tuke.sk}
\email{joe.ryan@newcastle.edu.au}
\email{thomas.kalinowski@newcastle.edu.au}

\begin{document}

\begin{abstract}
  A simple graph $G=(V,E)$ admits an $H$-covering if every edge in $E$ is contained in a subgraph
  $H'=(V',E')$ of $G$ which is isomorphic to $H$. In this case we say that $G$ is $H$-supermagic if
  there is a bijection $f:V\cup E\to\{1,\dotsc,\lvert V\rvert+\lvert E\rvert\}$ such that
  $f(V)=\{1,\dotsc,\lvert V\rvert\}$ and $\sum_{v\in V(H')}f(v)+\sum_{e\in E(H')}f(e)$ is constant
  over all subgraphs $H'$ of $G$ which are isomorphic to $H$. Extending results
  from~\cite{RoswithaBaskoro}, we show that the firecracker $F_{k,n}$ is
  $F_{2,n}$-supermagic, the banana tree $B_{k,n}$ is $B_{k-1,n}$-supermagic and the flower
  $\mathcal{F}_n$ is $C_3$-supermagic.

  \smallskip
  
\end{abstract}

\maketitle

\section{Introduction}

The graphs considered in this paper are finite, undirected and simple. For a positive integer $n$ we
denote the set $\{1,\dotsc,n\}$ by $[n]$, and for integers $a\leq b$, the set $\{a,\dotsc,b\}$ is
denoted by $[a,b]$. Let $V(G)$ and $E(G)$ be the set of vertices and edges of a graph $G$. A
\emph{graph labeling} is an assignment of integers to the vertices or edges, or both, subject to
certain conditions. Graph labeling was first introduced by Rosa~\cite{rosa} in 1967. Since then
there are various types of labeling that have been studied and developed (see~\cite{gallian}).

For a graph $H$, a graph $G$ admits an \emph{$H$-covering} if every edge of $G$ belongs to at least
one subgraph of $G$ which is isomorphic to $H$.  A graph $G=(V,E)$ which admits an $H$-covering is
called \emph{$H$-magic} if there exists a bijection
$f: V\cup E\to[\lvert V\rvert+\lvert E\rvert]$ and a constant $f(H)$, which we call the
\emph{$H$-magic sum} of $f$, such that $\sum_{v\in V(H')}f(v)+\sum_{e\in E(H')}f(e)=f(H)$ for every
subgraph $H'\subseteq G$ with $H'\cong H$. Additionaly, if $f(V)=[\lvert V\rvert]$
then we say that $G$ is \emph{$H$-supermagic}.

The concept of $H$-supermagic labeling was introduced by Guti\'{e}rrez and Llad\'{o}~\cite{gLlado}
in 2005, for $H$ being a star or a path. In~\cite{lladoMoragas}, Llad\'{o} and Moragas constructed
cycle-supermagic labelings for some graphs. Furthermore, Maryati et al.~\cite{MaryatiBS} studied
path-supermagic labelings while Ngurah et al.~\cite{ngurahSS}, Roswitha et al.~\cite{Chu} and
Kojima~\cite{Kojima} proved that some graphs have cycle-supermagic labelings. Some results for
certain shackles and amalgamations of a connected graph have been proved by Maryati et
al.~\cite{MaryatiSBRM}. Recently, Roswitha and Baskoro~\cite{RoswithaBaskoro} established
$H$-supermagic coverings for some trees.

Roswitha and Baskoro~\cite{RoswithaBaskoro} show that for any integer $k$ and even $n$, the
firecracker graph $F_{k,n}$ is $F_{2,n}$-supermagic and the banana tree graph $B_{k,n}$ is
$B_{(k-1),n}$-supermagic and left the remaining cases as open problems. In this paper, we solve
these two problems. The result for banana trees is an immediate consequence of a theorem about
amalgamations of graphs from~\cite{MaryatiSBRM} which we recall in Section~\ref{sec:bananatree}. The
result for firecrackers in Section~\ref{sec:PkGw} is obtained by a similar method. In addition, we
prove in Section~\ref{sec:flower} that for odd $n$, the flower graph $\mathcal{F}_n$ is
$C_3$-supermagic.

\section{Amalgamations and banana trees}\label{sec:bananatree}

Let $H$ be a graph with $n$ vertices, say $V(H)=\{v_1,\dotsc,v_n\}$ and $m$ edges, say
$E(H)=\{e_1,\dotsc,e_m\}$. Take $k$ copies of $H$ denoted by $H^1,\dotsc,H^k$ and let the vertex
and edge sets be $V(H^i) = \{v_1^i,\dotsc,v_n^i\}$ and $E(H^i)=\{e_1^i,\dotsc,e_m^i\}$. Fix a vertex
$v\in V(H)$, without loss of generality $v=v_n$, and form a graph, $G=A_k(H,v)$ by identifying all
the vertices $v_n^1,\dotsc,v_n^k$ (and denoting the identified vertex by $v_n$). The following
theorem was proved in~\cite{MaryatiSBRM}.
\begin{theorem}[\cite{MaryatiSBRM}]\label{thm:amalgamation}
  Let $H$ be any graph, and let $v\in V(H)$. If $G=A_k(H,v)$ contains exactly $k$ subgraphs
  isomorphic to $H$ then $G$ is $H$-supermagic with $H$-supermagic sums 
  \[f(H)=
    \begin{cases}
      \frac{3(n+m)-1}{2}+\frac{k(n+m-1)^2}{2} & \text{if }(m+n-1)(k-1)\text{ is even},\\
      \frac{3(n+m)-2}{2}+\frac{k\left[(n+m-1)^2+1\right]}{2} & \text{if }(m+n-1)(k-1)\text{ is odd}.
    \end{cases}
\]
\end{theorem}
For the convenience of the reader we provide an explicit description of the labeling.
\begin{proof}
 The graph $A_k(H,v)$ has $k(n-1)+1$ vertices and $km$ edges. We define the labeling
  \[f:V(A_k(H,v))\cup E(A_k(H,v))\to[k(n+m-1)+1]\]
  as follows.
  \begin{description}
  \item[Case 1.] If $n+m$ is odd, we start with $f(v_n)=1$. Then we use the labels
    $2,\dotsc,k(n-1)+1$ for the remaining vertices:
    \begin{align}
      f(v_i^j) &=
                 \begin{cases}
                   1+(i-1)k+j &\text{if }i\text{ is odd},\\
                   ik+2-j&\text{if }i\text{ is even},
                 \end{cases}
                              && \text{for }i\in[n-1],\,j\in[k].     \label{eq:vertices}
    \end{align}
    Finally we use the labels $k(n-1)+2,\dotsc,k(n+m-1)+1$ for the edges:
    \begin{align}
      f(e_i^j) &=
      \begin{cases}
        1 +(i+n-2)k+j & \text{if }i+n-1\text{ is odd},\\
        (i+n-1)k+2-j & \text{if }i+n-1\text{ is even},
      \end{cases}
      && \text{for }i\in[n-1],\,j\in[k].     \label{eq:edges}
    \end{align}
    The sum of the labels used for $H^j$ is independent of $j$:
    \begin{multline*}
      f(v_n)+\sum_{i=1}^{n-1}f(v_i^j)+\sum_{i=1}^{m}f(e_i^j)=1+\sum_{i=1,\,i\text{
          odd}}^{n+m-1}\left[1+(i-1)k+j\right]+\sum_{i=1,\,i\text{
          even}}^{n+m-1}\left[ik+2-j\right]\\
      =\frac{3(n+m)-1}{2}+\frac{k(n+m-1)^2}{2}.
    \end{multline*}
\item[Case 2.] If $n+m$ is even and $k$ is odd, we start with $f(v_n)=1$. Next we use the labels
  $2,\dotsc,3k+1$ to label the vertices $v_i^j$ for $i\in[3]$, $j\in[k]$ (assuming that $n\geq 4$,
  otherwise use the first edges in the obvious way):
  \begin{align*}
    f(v_1^j) &= 1+j \\
    f(v_2^j) &= 
               \begin{cases}
                 3(k+1)/2+j &\text{for }j\in[(k-1)/2],\\
                 (k+3)/2+j &\text{for }j\in[(k+1)/2,\,k],
               \end{cases}\\
    f(v_3^j) &=
               \begin{cases}
                 3k+2-2j &\text{for }j\in[(k-1)/2],\\
                 4k+2-2j &\text{for }j\in[(k+1)/2,\,k].
               \end{cases}               
  \end{align*}
  Then we use the labels $3k+2,\dotsc,k(n-1)+1$ for the remaining vertices, applying~(\ref{eq:vertices}) for
  $i\in[4,n-1]$. Finally, we use the labels $k(n-1)+2,\dotsc,k(n+m-1)+1$ for the edges,
  applying~(\ref{eq:edges}). The sum of the labels used for $H^j$ is independent of $j$:
  \begin{multline*}
    f(v_n) + \sum_{i=1}^3f(v_1^j)+ \sum_{i=4}^{n-1}f(v_i^j)+\sum_{i=1}^{m}f(e_i^j)\\ =
    1+\frac{9(k+1)}{2}+ \sum_{i=4,\,i\text{
          odd}}^{n+m-1}\left[1+(i-1)k+j\right]+\sum_{i=4,\,i\text{
          even}}^{n+m-1}\left[ik+2-j\right]\\
      =\frac{3(n+m)-1}{2}+\frac{k(n+m-1)^2}{2}.
  \end{multline*}
\item[Case 3.] If $n+m$ is even and $k$ is even, we start with $f(v_n)=k/2+1$. Next we use the labels
  $1,\dotsc,k/2,\,k/2+2,\dotsc,3k+1$ to label the vertices $v_i^j$ for $i\in[3]$, $j\in[k]$ (assuming that $n\geq 4$,
  otherwise use the first edges in the obvious way):
 \begin{align*}
    f(v_1^j) &=
               \begin{cases}
                 j & \text{for }j\in[k/2],\\
                 j+1& \text{for }j\in[k/2+1,\,k],
               \end{cases}\\
    f(v_2^j) &= 
               \begin{cases}
                 3k/2+1+j &\text{for }j\in[k/2],\\
                 k/2+1+j &\text{for }j\in[k/2+1,\,k],
               \end{cases}\\
    f(v_3^j) &=
               \begin{cases}
                 3(k+1)-2j &\text{for }j\in[k/2],\\
                 4k+2-2j &\text{for }j\in[k/2+1,\,k].
               \end{cases}               
  \end{align*}
  Then we use the labels $3k+2,\dotsc,k(n-1)+1$ for the remaining vertices, applying~(\ref{eq:vertices}) for
  $i=4,\dotsc,n-1$. Finally, we use the labels $k(n-1)+2,\dotsc,k(n+m-1)+1$ for the edges,
  applying~(\ref{eq:edges}). The sum of the labels used for $H^j$ is independent of $j$:
  \begin{multline*}
    f(v_n) + \sum_{i=1}^3f(v_1^j)+ \sum_{i=4}^{n-1}f(v_i^j)+\sum_{i=1}^{m}f(e_i^j)\\
    =(k/2+1)+ \frac{9k+8}{2} + \sum_{i=4,\,i\text{
          odd}}^{n+m-1}\left[1+(i-1)k+j\right]+\sum_{i=4,\,i\text{
          even}}^{n+m-1}\left[ik+2-j\right]\\
      =\frac{3(n+m)-2}{2}+\frac{k\left[(n+m-1)^2+1\right]}{2}. \qedhere
  \end{multline*}
  \end{description}
\end{proof}

Let $H$ be the graph obtained by taking a star with $n$ vertices and connecting an additional vertex
$v$ to exactly one leaf of the star. The \emph{banana tree} $B_{k,n}$ is the graph $A_k(H,v)$.
\begin{corollary}
  For any integers $k$ and $n\geq k+2$, the banana tree $B_{k,n}$ is $B_{1,n}$-supermagic.
\end{corollary}
The condition $n\geq k+2$ is needed because otherwise $B_{k,n}$ contains more than $k$ subgraphs
isomorphic to $B_{1,n}$. We do not have this problems for $H=B_{\ell,n}$ with $\ell\geq 2$, and
therefore we get the following result.
\begin{corollary}
  For any integers $n$, $k$ and $\ell\in[2,k-1]$, the banana tree $B_{k,n}$ is
  $B_{\ell,n}$-supermagic. In particular, for $\ell=k-1$, this solves the open problem
  in~\cite{RoswithaBaskoro}.
\end{corollary}

\begin{remark}
   Note that the labeling strategy in the first case of the proof of Theorem~\ref{thm:amalgamation}
  immediately gives the following result. Fix an induced subgraph $H'$ of $H$, say induced by the
  last $\ell$ vertices, and form a graph, $G=A_k(H,H')$ by identifying the vertices
  $v_i^1,\dotsc,v_i^k$ for $i=n-\ell+1,\dotsc,n$. If $n-\ell+\lvert E(H)\setminus E(H')\rvert$
  is even and $G$ contains exactly $k$ subgraphs isomorphic to $H$, then $G$ is $H$-supermagic.
\end{remark}

\section{Attaching copies of a fixed graph to a path}\label{sec:PkGw}

Let $G$ be a graph with $n$ vertices, say $V(G)=\{v_1,\dotsc,v_n\}$ and $m$ edges, say
$E(G)=\{e_1,\dotsc,e_m\}$. Let $P_k$, $k\geq2$, be a path with vertex set
$V(P_k)=\{w_1,w_2,\dotsc,w_k\}$ and edge set $E(P_k)=\{w_1w_2,\dotsc,w_{k-1}w_k\}$. Take $k$ copies
of $G$ denoted by $G^1,G^2,\dotsc,G^k$ and let the vertex and edge sets be
$V(G_i) = \{v_1^i,\dotsc,v_n^i\}$ and $E(G^i)=\{e_1^i,\dotsc,e_m^i\}$. Fix a vertex $v\in V(G)$,
without loss of generality $v=v_n$, and attach the copies of $G$ to the path such that the vertex
$v_n^i\in V(G^i)$ is identified with the vertex $w_i$ in $P_k$, $i=1,2,\dotsc,k$. The resulting
graph is denoted by $P_k(G,v)$.

\begin{theorem}\label{thm:supermagic_labeling}
  Let $G$ be graph with $n$ vertices and $m$ edges and let $k\geq 2$ be an integer. If $(n+m-1)(k-1)$ is even
  and $P_k(G,v)$ contains exactly $k-1$ subgraphs isomorphic to $P_2(G,v)$, then $P_k(G,v)$ is
  $P_2(G,v)$-supermagic with supermagic sum $(n+m)[(n+m+1)k+1]+\left\lceil k/2\right\rceil$.
\end{theorem}
\begin{proof}
  The graph $P_k(G,v)$ has $kn$ vertices and $(m+1)k-1$ edges. We define the labeling
  \[f:V(P_k(G,v))\cup E(P_k(G,v))\to[(m+n+1)k-1]\]
  as follows. For the path vertices, we use the first $k$ labels $1,\dotsc,k$:
    \[f(w_i)=
      \begin{cases}
        (i+1)/2 & \text{if $i$ is odd},\\
        \left\lceil k/2\right\rceil + i/2 & \text{if $i$ is even},
      \end{cases} \qquad\text{for }i\in[k].
    \]
  For the path edges, we use the the labels in $[(m+n)k+1,(m+n)k+(k-1)]$:
    \[f(w_iw_{i+1})=(n+m+1)k-i\qquad\text{for $i\in[k-1]$}.\]  
  For labeling the remaining elements we distinguish two cases.
  \begin{description}
  \item[Case 1.] If $n+m-1$ is even, we set
    \begin{align}
      \label{eq:vertices_2}
      f\left(v^i_j\right) &=
      \begin{cases}
        jk +i & \text{if $j$ is odd,}\\
        (j+1)k+1-i& \text{if $j$ is even,}\\
      \end{cases} && \text{for }j\in[n-1],\,i\in[k],\\     
    f\left(e^i_j\right) &=
      \begin{cases}
        (n-1+j)k +i & \text{if $n-1+j$ is odd,}\\
        (n+j)k+1-i& \text{if $n-1+j$ is even,}\\
      \end{cases} && \text{for }j\in[m],\,i\in[k].\label{eq:edges_2}
    \end{align}
  \item[Case 2.] If $n+m-1$ is odd and $k$ is odd, we set
    \begin{align*}
     f(v_1^i) &= k+i \\
    f(v_2^i) &= 
               \begin{cases}
                 (5k+1)/2+i &\text{for }i\in[(k-1)/2],\\
                 (3k+1)/2+i &\text{for }i\in[(k+1)/2,\,k],
               \end{cases}\\
    f(v_3^i) &=
               \begin{cases}
                 4k+1-2i &\text{for }i\in[(k-1)/2],\\
                 5k+1-2i &\text{for }i\in[(k+1)/2,\,k].
               \end{cases}
    \end{align*}
  \end{description}
  As in the proof of Theorem \ref{thm:amalgamation}, (\ref{eq:vertices_2}) and (\ref{eq:edges_2})
  are used for labeling the remaining vertices and edges. Denoting the sum of the labels used for
  $G^i$ by $A_i$, we obtain
  \begin{equation*}
    A_i = \sum_{j=1}^{n-1}f(v_j^i)+\sum_{j=1}^{m}f(e_j^i)+f(w_i)
    =\frac{(n+m-1)(n+m+1)k+(n+m-1)}{2}+
    \begin{cases}
      (i+1)/2 & \text{if $i$ is odd},\\
      \left\lceil k/2\right\rceil + i/2 & \text{if $i$ is even}.
    \end{cases}
  \end{equation*}        
  Finally, the sum of the labels of the subgraph isomorphic to $P_2(G,v)$ which is formed by $G^i$,
  $G^{i+1}$ and the edge $w_iw_{i+1}$ is independent of $i$:
  \[A_i+A_{i+1}+f(w_iw_{i+1})=(n+m)[(n+m+1)k+1]+\left\lceil k/2\right\rceil.\qedhere\]
\end{proof}
\begin{remark}
  We think that it might be possible that the parity assumption in
  Theorem~\ref{thm:supermagic_labeling} is not necessary, and we leave the case that both $n+m$ and
  $k$ are even for future work.
\end{remark}
\begin{example}\label{ex:1}
  We illustrate the construction in Theorem~\ref{thm:supermagic_labeling} for $k=5$, $G=K_4^-$ (the
  graph obtained from a complete graph on $4$ vertices by deleting one edge) and $v$ being a vertex
  of degree $3$ in $G$. We obtain the $P(K_4^-,v_4)$-supermagic labeling shown in
  Figure~\ref{fig:labeling}.
  	\begin{figure}[htb]
		\centering
		\begin{tikzpicture}[xscale=1.24,yscale=1.55,every node/.style={draw,shape=circle,outer sep=2pt,inner sep=2pt,minimum size=.5cm}]
			\scriptsize
			\foreach \j in {1,...,3}{
				\pgfmathtruncatemacro{\jk}{2*\j-1}
				\node (v\jk) at (5*\j-5,0) 	{\j};
				\pgfmathtruncatemacro{\label}{4+2*\j}
				\node (w\jk1) at (5*\j-6.2,1.3)	{\label};
				\pgfmathtruncatemacro{\label}{17-2*\j}
				\node (w\jk2) at (5*\j-5,2)	{\label};
				\pgfmathtruncatemacro{\label}{14+2*\j}
				\node (w\jk3) at (5*\j-3.8,1.3)	{\label};
				\pgfmathtruncatemacro{\label}{24+2*\j}
				\draw[thick] (v\jk) to node[draw=none,fill=none,below,near end] {\label} (w\jk1);
				\pgfmathtruncatemacro{\label}{34+2*\j}
				\draw[thick] (v\jk) to node[draw=none,fill=none,left=-4pt] {\label} (w\jk2);
				\pgfmathtruncatemacro{\label}{27-2*\j}
				\draw[thick] (v\jk) to node[draw=none,fill=none,below,near end] {\label} (w\jk3);
				\pgfmathtruncatemacro{\label}{37-2*\j}
				\draw[thick] (w\jk1) to node[draw=none,fill=none,above] {\label} (w\jk2);
				\pgfmathtruncatemacro{\label}{47-2*\j}
				\draw[thick] (w\jk2) to node[draw=none,fill=none,above] {\label} (w\jk3);
			}
			\foreach \j in {1,2}{
				\pgfmathtruncatemacro{\jk}{2*\j}
				\pgfmathtruncatemacro{\label}{\j+3}
				\node (v\jk) at (5*\j-2.5,0) 	{\label};
				\pgfmathtruncatemacro{\label}{5+2*\j}
				\node (w\jk1) at (5*\j-3.7,-1.3)	{\label};
				\pgfmathtruncatemacro{\label}{16-2*\j}
				\node (w\jk2) at (5*\j-2.5,-2)	{\label};
				\pgfmathtruncatemacro{\label}{15+2*\j}
				\node (w\jk3) at (5*\j-1.3,-1.3)	{\label};
				\pgfmathtruncatemacro{\label}{25+2*\j}
				\draw[thick] (v\jk) to node[draw=none,fill=none,above,near end] {\label} (w\jk1);
				\pgfmathtruncatemacro{\label}{35+2*\j}
				\draw[thick] (v\jk) to node[draw=none,fill=none,left=-4pt] {\label} (w\jk2);
				\pgfmathtruncatemacro{\label}{26-2*\j}
				\draw[thick] (v\jk) to node[draw=none,fill=none,above,near end] {\label} (w\jk3);
				\pgfmathtruncatemacro{\label}{36-2*\j}
				\draw[thick] (w\jk1) to node[draw=none,fill=none,below] {\label} (w\jk2);
				\pgfmathtruncatemacro{\label}{46-2*\j}
				\draw[thick] (w\jk2) to node[draw=none,fill=none,below] {\label} (w\jk3);
			}
			\foreach \j in {1,...,4}{
				\pgfmathtruncatemacro{\next}{\j+1}
				\pgfmathtruncatemacro{\label}{50-\j}
				\draw[thick] (v\j) to node[draw=none,fill=none,below=-4pt] {\label} (v\next);
			}
		\end{tikzpicture}
		\caption{$P_2(K_4^-,v_4)$-supermagic labeling for $P_5(K_4^-,v_4)$.}\label{fig:labeling}
	\end{figure}
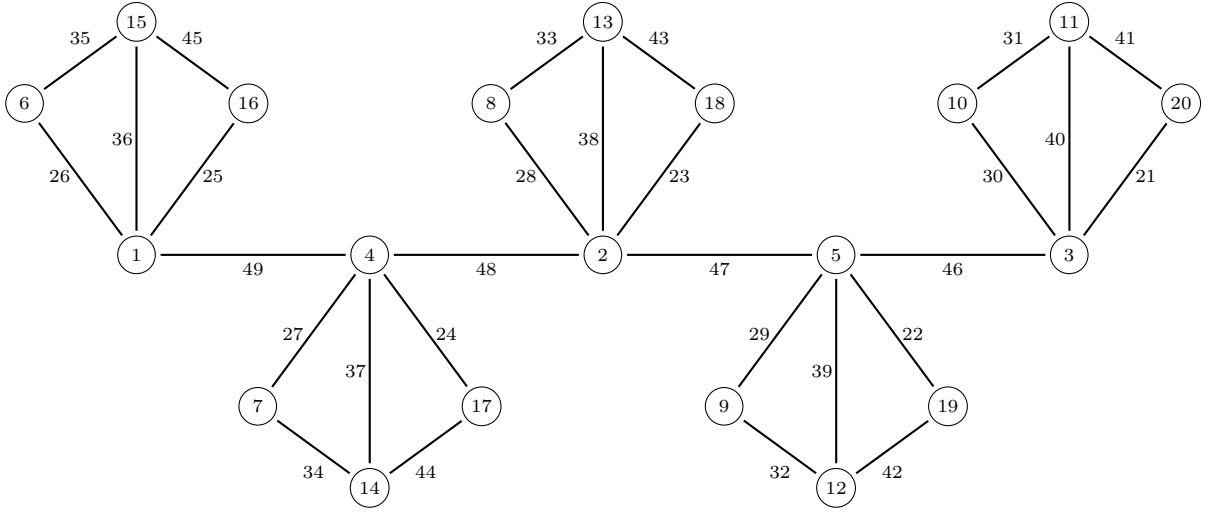
\end{example}

\begin{corollary}
  Let $G=K_{1,n-1}$ be a star with $n\geq 4$ vertices, and let $v$ be a pendant vertex of $G$. The
  firecracker graph is $F_{k,n}=P_k(G,v)$. Since $\lvert V(G)\rvert+\lvert E(G)\rvert=2n-1$ is odd,
  and there are exactly $k-1$ subgraphs isomorphic to $F_{2,n}$, the firecracker $F_{k,n}$ is
  $F_{2,n}$-supermagic with supermagic sum $(2n-1)(2nk+1)+\left\lceil k/2\right\rceil$.
\end{corollary}

\section{$C_3$-Supermagic Labeling of the Flower Graph $\mathcal{F}_n$}\label{sec:flower}

A \emph{flower} graph $\mathcal{F}_n$ is constructed from a wheel $W_n$ by adding $n$ vertices, each
new vertex adjacent to one vertex on the cycle and the center of the wheel with vertex set
$V=\{x_0\}\cup\{x_i:1\leq i\leq n\}\cup\{y_i:1\leq i\leq n\}$ and edge set
$E=\{x_0x_i\,:\,1\leq i\leq n\}\cup\{x_0y_i\,:\,1\leq i\leq n\}\cup\{x_iy_i\,:\,1\leq i\leq
n\}\cup\{x_ix_{i+1}\,:\,1\leq i\leq n\},$ where indices are interpreted modulo $n$ in the obvious
way.

We consider four permutations $\pi_1,\dotsc,\pi_4$ of the set $[n]$, and define a total labeling of
the flower graph $\mathcal{F}_n$ as follows.
\begin{align*}
f(x_0)	&=	n+1,\\
f(x_i)	&=	\pi_1(i)	&& \text{for }i\in[n],\\
f(y_i)	&=	\pi_2(i)+n+1	&& \text{for }i\in[n],\\
f(x_0x_i)	&= 	\pi_2(i)+5n+1	&& \text{for }i\in[n],\\
f(x_0y_i)	&= 	\pi_2(i)+4n+1	&& \text{for }i\in[n],\\
f(x_iy_i)	&= 	\begin{cases}
\pi_3(i)+2n+1 & \hbox{odd $i$}\\
\pi_3(i)+3n+1 & \hbox{even $i$}
\end{cases}	&& \text{for }i\in[n],\\
f(x_ix_{i+1}) &= \pi_4(i)+2n+1+\frac{n+1}{2} && \text{for }i\in[n-1].
\end{align*}

\begin{lemma}\label{lem:sumflower}
  Define
  \begin{align*}
    \varphi_1^i(\pi_1,\dotsc,\pi_4)	&= \pi_1(i)+\pi_1(i+1)+\pi_2(i)+\pi_2(i+1)+\pi_4(i)+\frac{n+1}{2}-1,\\
    \varphi_2^i(\pi_1,\dotsc,\pi_4)	&= \pi_1(i)+3\pi_2(i)+\pi_3(i)+
                                          \begin{cases}
                                            n & \text{if $i$ is even},\\
                                            0 & \text{if $i$ is odd}.
                                          \end{cases}
  \end{align*}
  If $\varphi_k^i(\pi_1,\dotsc,\pi_4)$ is equal to a constant $\varphi$ for all $i\in[n]$ and
  $k\in\{1,2\}$, then the labeling given above is $C_3$-supermagic with supermagic sum
  $f(C_3)=13n+5+\varphi.$
\end{lemma}

\begin{proof}
  The flower graph $F_n$ contains $2n$ subgraphs $H_1, \dotsc, H_{2n}$ isomorphic to $C_3$. We
  distinguish two types of 3-cycles: (1) cycles induced by vertex sets $\{x_0,\,x_i,\,x_{i+1}\}$,
  and (2) cycles induced by vertex sets $\{x_0,\,x_i,\,y_i\}.$
  \begin{description}
  \item[Case 1.] Cycle $(x_0,x_i,x_{i+1})$. The sum of the vertex labels is
    \begin{equation}\label{eq:c1vertex}
      f(x_0)+f(x_i)+f(x_{i+1})=n+1+\pi_1(i)+\pi_1(i+1),
    \end{equation}
    and the sum of the edge labels is
    \begin{equation}\label{eq:c1edge}
      f(x_0x_i)+f(x_ix_{i+1})+f(x_0x_{i+1})=\pi_2(i)+\pi_2(i+1)+\pi_4(i)+12n+3+\frac{n+1}{2}.
    \end{equation}
    Taking the sum of~(\ref{eq:c1vertex}) and~(\ref{eq:c1edge}), we have the supermagic sum
    \begin{align*}
      f(C_3)	&=	n+1+\pi_1(i)+\pi_1(i+1)+\pi_2(i)+\pi_2(i+1)+\pi_4(i)+12n+3+\frac{n+1}{2}\\
                &= 13n+5+\left[\pi_1(i)+\pi_1(i+1)+\pi_2(i)+\pi_2(i+1)+\pi_4(i)-1+\frac{n+1}{2}\right]\\
                &= 13n+5+\varphi_1^i(\pi_1,\dotsc,\pi_4) = 13n+5+\varphi.
    \end{align*}
  \item[Case 2.] Cycle $(x_0,x_i,y_i)$. The sum of the vertex labels is
    \begin{equation}\label{eq:c2vertex}
      f(x_0)+f(x_i)+f(y_i)=2n+2+\pi_1(i)+\pi_2(i),
    \end{equation}
    and the sum of the edge labels is
    \begin{equation}\label{eq:c2edge}
      f(x_0x_i)+f(x_iy_i)+f(x_0y_i)=\pi_2(i)+\pi_2(i)+\pi_3(i)+
      \begin{cases}
        11n+3 &\text{if $i$ is odd},\\
        12n+3 &\text{if $i$ is even}.
      \end{cases}
      	\end{equation}
    Taking the sum of~(\ref{eq:c2vertex}) and~(\ref{eq:c2edge}) gives the supermagic sum
    \[f(C_3) = 13n+5+\varphi_2^i(\pi_1,\dotsc,\pi_4)= 13n+5+\varphi.\qedhere\]
  \end{description}
\end{proof}

In the following lemma we provide permutations $\pi_1,\dotsc,\pi_4$ which satisfy the condition in
Lemma~\ref{lem:sumflower}.

\begin{lemma}\label{lem:permutFlower}
  Define the permutations by
  \begin{align*}
    \pi_1(i)	&=	i,\\
    \pi_2(i)	&=	n+1-
                  \begin{cases}
                    (i+1)/2				& \hbox{for odd $i$,}\\
                    i/2+(n+1)/2 & \hbox{for even $i$,}
                  \end{cases}\\
    \pi_3(i) &=
               \begin{cases}
                 (i+1)/2				& \hbox{for odd $i$,}\\
                 i/2+(n+1)/2 & \hbox{for even $i$,}
               \end{cases}\\
    \pi_4(i) &= n+1-i.
  \end{align*}
  Then for every $i\in[n]$ we have
  $\varphi^i_1(\pi_1,\dotsc,\pi_4)=\varphi^i_2(\pi_1,\dotsc,\pi_4)=3n+2.$
\end{lemma}

\begin{theorem}
  For any odd integer $n$, the flower graph $\mathcal{F}_n$ is $C_3$-supermagic.
\end{theorem}
\begin{proof}
  The total labeling of $\mathcal{F}_n$ can be obtained by applying the permutations in Lemma~\ref{lem:permutFlower} to
  the labeling construction. Using the value of $\varphi$ in Lemma~\ref{lem:permutFlower} and the
  supermagic the sum of the permutations in Lemma~\ref{lem:permutFlower}, we have the constant
  supermagic sum on flower graph $\mathcal{F}_n$
  \[  f(C_3)	= 13n+5+\varphi= 13n+5+3n+2= 16n+7.\]  
  Hence, flower graph $F_n$ is $C_3$-supermagic, for odd $n$.
\end{proof}

\begin{example}
  Using Lemma \ref{lem:permutFlower} for $n=7$ we get the permutations $\pi_1=(1,2,3,4,5,6,7)$,
  $\pi_2=(7,3,6,2,5,1,4)$, $\pi_3=(1,5,2,6,3,7,4)$ and $\pi_4=(7,6,5,4,3,2,1)$. These permutations
  give the labeling for $\mathcal{F}_7$ shown in Figure~\ref{fig:flower7}.

\begin{figure}[htb]
	\centering
	\begin{tikzpicture}[scale=1.8,every node/.style={draw,shape=circle,outer sep=2pt,inner sep=1pt,minimum
			size=.4cm}]
		\small
		\node (v0) at (0,0) {$8$};
		
		\node (v1) at (0,1) 	{$1$};
		\node (v2) at (.8,.64) 	{$2$};
		\node (v3) at (1,-.24) 	{$3$};
		\node (v4) at (.44,-.92) {$4$};
		\node (v5) at (-.44,-.92) {$5$};
		\node (v6) at (-1,-.24) {$6$};
		\node (v7) at (-.8,.64)	{$7$};
		
		\node (u1) at (0,2) 	{$15$};
		\node (u2) at (1.6,1.22) {$11$};
		\node (u3) at (1.98,-.48) {$14$};
		\node (u4) at (.86,-1.84) {$10$};
		\node (u5) at (-.86,-1.84){$13$};
		\node (u6) at (-1.98,-.48) {$9$};
		\node (u7) at (-1.6,1.22) {$12$};
		
		\draw[thick] (v0) to node[draw=none,right=3pt,above=-5pt,sloped] {$43$} (v1);
		\draw[thick] (v0) to node[draw=none,right=3pt,above=-5pt,sloped] {$39$} (v2);
		\draw[thick] (v0) to node[draw=none,right=3pt,above=-5pt,sloped] {$42$} (v3);
		\draw[thick] (v0) to node[draw=none,right=3pt,above=-4pt,sloped] {$38$} (v4);
		\draw[thick] (v0) to node[draw=none,right=-3pt,below=-4pt,sloped] {$41$} (v5);
		\draw[thick] (v0) to node[draw=none,right=-3pt,below=-4pt,sloped] {$37$} (v6);
		\draw[thick] (v0) to node[draw=none,right=-3pt,below=-4pt,sloped] {$40$} (v7);
		
		\draw[thick] (v1) to node[draw=none,left=-3pt,above=-4pt,sloped] {$26$} (v2);
		\draw[thick] (v2) to node[draw=none,left=-3pt,above=-4pt,sloped] {$25$} (v3);
		\draw[thick] (v3) to node[draw=none,left=3pt,below=-4pt,sloped] {$24$} (v4);
		\draw[thick] (v4) to node[draw=none,left=3pt,below=-4pt,sloped] {$23$} (v5);
		\draw[thick] (v5) to node[draw=none,left=3pt,below=-4pt,sloped] {$22$} (v6);
		\draw[thick] (v6) to node[draw=none,left=-3pt,above=-4pt,sloped] {$21$} (v7);
		\draw[thick] (v7) to node[draw=none,left=-3pt,above=-4pt,sloped] {$20$} (v1);
		
		\draw[thick] (v1) to node[draw=none,left=-3pt] {$16$} (u1);
		\draw[thick] (v2) to node[draw=none,above=-4pt,sloped] {$27$} (u2);
		\draw[thick] (v3) to node[draw=none,above=-4pt,sloped] {$17$} (u3);
		\draw[thick] (v4) to node[draw=none,above=-4pt,sloped] {$28$} (u4);
		\draw[thick] (v5) to node[draw=none,below=-4pt,sloped] {$18$} (u5);
		\draw[thick] (v6) to node[draw=none,below=-4pt,sloped] {$29$} (u6);
		\draw[thick] (v7) to node[draw=none,below=-4pt,sloped] {$19$} (u7);
		
		\draw[thick] (v0) to[out=360/7+12,in=6*360/7-8] 	node[draw=none,right=4pt,above=5pt] {$36$} (u1);
		\draw[thick] (v0) to[out=12,in=5*360/7-8]			node[draw=none,left=-10pt,below=-5pt,sloped] {$32$} (u2);
		\draw[thick] (v0) to[out=6*360/7+12,in=4*360/7-8] 	node[draw=none,left=-10pt,below=-5pt,sloped] {$35$} (u3);
		\draw[thick] (v0) to[out=5*360/7+12,in=3*360/7-8] 	node[draw=none,left=-10pt,below=-5pt,sloped] {$31$} (u4);
		\draw[thick] (v0) to[out=4*360/7+12,in=2*360/7-8]	node[draw=none,right=-12pt,above=-6pt,sloped] {$34$} (u5);
		\draw[thick] (v0) to[out=3*360/7+12,in=360/7-8]		node[draw=none,right=-12pt,above=-6pt,sloped] {$30$} (u6);
		\draw[thick] (v0) to[out=2*360/7+12,in=-8] 			node[draw=none,right=-12pt,above=-6pt,sloped] {$33$} (u7);
	\end{tikzpicture}
	\caption{$C_3$-supermagic labeling of $\mathcal{F}_7$ with supermagic sum $f(C_3)=119$.}\label{fig:flower7}
\end{figure}
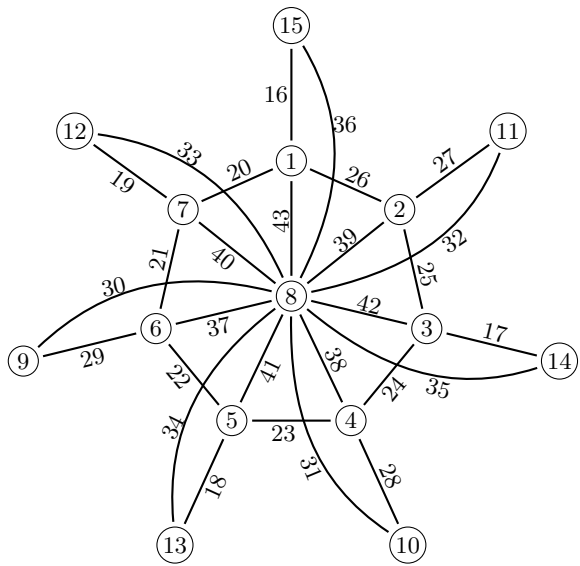
\end{example}

\subsubsection*{Acknowledgement}
The first author was supported by the Indonesian government under the Indonesia Endowment Fund for
Education (LPDP) Scholarship.  The second author was supported by APVV-15-0116 and
by VEGA 1/0385/17.


\providecommand{\bysame}{\leavevmode\hbox to3em{\hrulefill}\thinspace}
\providecommand{\MR}{\relax\ifhmode\unskip\space\fi MR }
\providecommand{\MRhref}[2]{%
  \href{http://www.ams.org/mathscinet-getitem?mr=#1}{#2}
}
\providecommand{\href}[2]{#2}

\end{document}